\documentclass[letterpaper,10pt,conference]{ieeeconf}
\IEEEoverridecommandlockouts
\usepackage{amsmath, graphicx, epsfig, color, amsfonts, algorithm, amssymb, mathtools}
\usepackage{color,soul}

\usepackage[font=footnotesize]{caption}

\usepackage{subcaption}
\usepackage{float}

\graphicspath{{./figures/}}

\usepackage{version,xspace}
\usepackage[noend]{algorithmic}

\def\real{\mathbb{R}}

\newcommand{\until}[1]{\{1,\dots, #1\}}
\newcommand{\subscr}[2]{#1_{\textup{#2}}}
\newcommand{\supscr}[2]{#1^{\textup{#2}}}
\newcommand{\setdef}[2]{\{#1 \; | \; #2\}}

\newcommand{\map}[3]{#1: #2 \rightarrow #3}

\newcommand\oprocendsymbol{\hbox{$\square$}}
\newcommand\oprocend{\relax\ifmmode\else\unskip\hfill\fi\oprocendsymbol}

\renewcommand{\theenumi}{\roman{enumi}}

\newcommand\bit[1]{\textit{\textbf{#1}}}

\def \bs {\boldsymbol}
\def \mc {\mathcal}

\def\tran{^\top}

\def \etal {\emph{et al.}}

\def \mr {\mathrm}

\renewcommand{\baselinestretch}{0.985}

\newtheorem{theorem}{Theorem}

\newtheorem{lemma}{Lemma}

\newtheorem{remark}{Remark}

\newtheorem{assumption}{Assumption}
\newtheorem{definition}{Definition}

\renewcommand{\theenumi}{(\roman{enumi}}

\usepackage{array}

\newcommand{\norm}[1]{\left\lVert#1\right\rVert}

\title{Towards Modeling Human Motor Learning Dynamics \\ in High-Dimensional Spaces
\thanks{This work has been supported by NSF Award CMMI-1940950.}
}

\author{Ankur Kamboj, Rajiv Ranganathan, Xiaobo Tan, and Vaibhav Srivastava\\
\thanks{A. Kamboj, X. Tan, and V. Srivastava are with the Department of Electrical and Computer Engineering, Michigan State University, East Lansing, MI 48823 USA. 
\texttt{\{ankurank, xbtan, vaibhav\}@msu.edu}.}
\thanks{R. Ranganathan is with the Department of Kinesiology, Michigan State University, East Lansing, MI 48823 USA. 
\texttt{rrangana@msu.edu}.}
}

\begin{document}


\maketitle

\begin{abstract}
Designing effective rehabilitation strategies for upper extremities, particularly hands and fingers, warrants the need for a computational model of human motor learning. The presence of large degrees of freedom (DoFs) available in these systems makes it difficult to balance the trade-off between learning the full dexterity and accomplishing manipulation goals. The motor learning literature argues that humans use motor synergies to reduce the dimension of control space. Using the low-dimensional space spanned by these synergies, we develop a computational model based on the internal model theory of motor control. We analyze the proposed model in terms of its convergence properties and fit it to the data collected from human experiments. We compare the performance of the fitted model to the experimental data and show that it captures human motor learning behavior well.
\end{abstract}

\IEEEpeerreviewmaketitle

\section{Introduction}
Motor impairments after a neurological injury such as stroke are one of the leading causes of disability in the United States \cite{mozaffarian2015executive}. Impairments of upper extremities in particular, including hand and finger function, are common, with $75\%$ of stroke survivors facing difficulties performing daily activities \cite{g1999long,lawrence2001estimates}. Critically, impairments after stroke include not only muscle and joint-specific deficits such as weakness and changes in the kinematic workspace \cite{cruz2005kinetic}, but also coordination deficits such as reduced independent joint control and impairments in finger individuation and enslaving \cite{lang2003differential,li2003effects}. Thus, understanding how to address these coordination deficits is critical for improving movement rehabilitation.

A key step in addressing these coordination deficits in the upper extremity and promoting rehabilitation is the development of a computational model of learning that takes place in human motor systems \cite{zhou2016learning}. These motor systems involve large DoFs; for example, the hand alone has over 25 joints and 30 muscles, which gives rise to the incredible dexterity that we possess when grasping and manipulating objects. A key challenge to motor learning in such high-dimensional spaces is to balance the exploration required to learn full dexterity induced by the high DoFs and learning just enough dexterity to generate desired motion. The motor learning literature has argued that the human nervous system uses a small number of \emph{motor synergies}, which are defined as a coordinated motion of a group of joints/DoFs, to control the high-DoF body motor systems \cite{tresch2006matrix}. Santello \etal~\cite{santello1998postural} concluded that two principal postural synergies accounted for more than $80\%$ of the postural data collected from various grasp postures, suggesting that most of the grasping postures can be explained by just two underlying motor synergies, a substantial reduction from the 15 DoF data that they recorded. We build on this idea of motor synergies to develop a dynamical model of human motor learning (HML).

The need for HML models for physical as well as robotic rehabilitation has been outlined in \cite{zhou2016learning}.
While several learning mechanisms have been proposed for HML processes \cite{haith2013model}, it has mostly been viewed in the literature as an adaptation in motor tasks, and consequently, several computational models of HML have been rooted in adaptive control theory. For example, Shadmehr \etal~\cite{shadmehr1994adaptive} proposed the idea of dynamic internal model in HML that adapts to the changed limb dynamics because of a force field imposed on the hand. Similarly, Zhou \etal~\cite{zhou2011human} leveraged iterative model reference adaptive control to model human learning through repetitive tasks.
However, most of these works only deal with motor learning in relatively low-dimensional spaces.

Pierella \etal~\cite{pierella2019dynamics} study motor learning of a novel task in high DoF motor systems where multiple solutions are available. They develop a dynamic model of motor learning that relies on the formation of internal models only at the end of a trial. We develop similar models for motor learning tasks studied in this paper; however, this work is distinct from \cite{pierella2019dynamics} in two key aspects: (a) the motor learning tasks studied in this paper involve higher DoFs, which require the use of synergies for efficient modeling, and (b) our setup deals with continuous feedback, which leads to faster learning but presents modeling and estimation challenges.

In this paper, we develop a computational model of HML in high-dimensional spaces using the context of learning novel coordinated movements of hand finger joints. Our model is based on the internal model theory of motor control~\cite{shadmehr2010error} and leverages the idea of motor synergies for efficient learning using low-dimensional representations. Based on the evidence in the motor learning literature~\cite{yavari2013fast}, the model comprises fast forward learning dynamics and slow inverse learning dynamics. We analyze the model using techniques from adaptive control theory and singular perturbation analysis, and establish its convergence properties. Finally, we fit the proposed model to experimental data and compare the fitted model with the experimental data to show that it captures human motor learning behavior effectively.

The rest of the paper is organized as follows. Section \ref{section2} outlines the experimental setup that we leverage to build our motor learning model as well as some data preprocessing techniques we adopted in the paper. Section \ref{section3} describes the proposed computational model and is followed by its convergence analysis in Section \ref{section4}. In Section \ref{section5}, we fit our model to experimental data and evaluate its performance with the human data. We conclude and provide some future directions in Section \ref{section6}.

\begin{figure}[h!]
    \centering
    \includegraphics[width=0.6\linewidth]{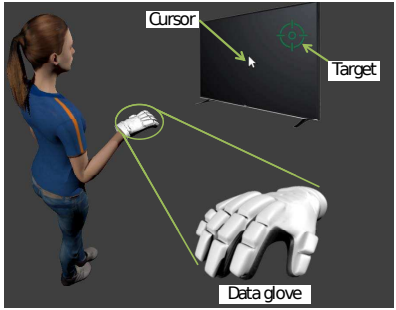}
    \caption{Artistic rendering of the experimental setup with the subject wearing the data glove.}
    \label{fig:experiment}
\end{figure}

\section{Background \& Experimental Setup} \label{section2}
In this section, we present the experimental paradigm leveraged in this paper to model HML in high-dimensional spaces. We also discuss techniques for preprocessing of data collected from these experiments.

Motor (re)-learning is central to the rehabilitation of individuals with hand and finger impairments and/or paresis. It is generally accepted that the mechanisms underlying re-learning during rehabilitation are governed by general principles of motor learning based on the learning of novel tasks by healthy subjects \cite{muratori2013applying, krakauer2006motor}. We focus on the motor learning experiment presented in \cite{mosier2005remapping} wherein healthy subjects learn a novel motor task. Each subject wears a data glove, which records the movements of the 19 finger joints. A body-machine interface (BoMI) uses a matrix to project the 19-dimensional finger movements onto the movement of a cursor on a 2-D computer screen; see Fig.~\ref{fig:experiment}. Specifically, the BoMI projects finger joint velocities $\bs u \in \real^m$ to cursor velocity $\dot{\bs x} \in \real^n$ using a matrix $C \in \real^{n\times m}$ such that
\begin{align}\label{base_dyn}
    \dot{\bs x} = C \bs u.    
\end{align}
Here, $n=2$ and $m=19$. Let $\bs q \in \real^m$ be the vector of finger joint angles and thus, $\dot{\bs q} = \bs u$. 
The subjects need to move the cursor to a target point, and a new target point is prescribed once the current target is reached. To this end, subjects need to learn coordinated finger joint movements that are consistent with the (unknown) projection matrix $C$. This experiment naturally involves motor learning in high-dimensional spaces, with the output being in the low-dimensional screen space. Each subject participated in eight sessions, each comprised of 60 trials. In each trial, the sequence of target points were randomly selected from 4 points located at $(0.5,4.5), (2.5, 0.5), (2.5, 2.5), (4.5, 4.5)$ units on the screen. Due to the large nullspace of $C$, this task is redundant in the sense that a desired cursor movement can be achieved with multiple synergistic motions of finger joints.
This makes learning the optimal (minimum energy) synergistic motion of finger joints particularly challenging.

The data glove records finger joint angles which are used to update the cursor position. Therefore, using standard techniques from adaptive control \cite{sastry1990adaptive}, we write (\ref{base_dyn}), in terms of \emph{filtered increments in joint angles} $\delta \bs q$ and \emph{filtered increments cursor positions} $\delta \bs x$, as
\begin{equation} \label{system}
    \delta \bs x = C \delta \bs q.
\end{equation}
Here, the evolution of $\delta \bs x$ and $\delta \bs q$ is governed by the dynamic equations
\begin{subequations} \label{filtered_pos}
\begin{align}
\dot{\bs \chi} &= - a\bs \chi + \bs x \\
    \delta \bs x &= - a\bs \chi + \bs x \\
    \dot{\delta \bs q} &= -a \delta \bs q + \bs u, \label{delta_u}
\end{align}
\end{subequations}
where $a$ is sufficiently large.

\section{Adaptive Control-based HML Computational Model} \label{section3}
Motor learning literature \cite{tresch2006matrix, santello1998postural} suggests that humans control the high DoF body motor systems using a small number of coordinated joint movement patterns called synergies. For example, four synergies spanned more than $80\%$ of the finger joint configurations required to generate the American Sign Language (ASL) Alphabet~\cite{vinjamuri2014candidates, liu2008contributions}.
We assume that the mapping matrix $C=W\Phi$, where $\Phi \in \real^{h\times m}$ is a matrix of $h$ basic synergies underlying coordinated human finger motions, and $W \in \real^{n \times h}$ represents the contributions (weights) of these synergies during a particular hand motion. We assume that $\Phi$ is an orthonormal matrix, i.e., the synergies contributing to the hand motions lie in orthogonal spaces. We further assume these synergies to be known (see Section \ref{section5.1.1} for details). While our motor learning model is in a high-dimensional space, these synergies reduce the size of the learning space and enable efficient learning by reducing the amount of desired exploration.

It is believed that motor learning in humans involves the formation of internal models \cite{shadmehr1994adaptive}, including models for forward and inverse learning. The forward model predicts the outcomes of motor actions (hand joint movements), whereas the inverse model 
predicts the motor actions required to generate certain outcomes.
In this section, we adopt techniques from adaptive control theory to develop models of forward and inverse HML dynamics.

\subsection{Adaptive Control-based Model of Forward Learning}
In the context of the experimental setup described in Section \ref{section2}, forward learning entails learning the forward BoMI mapping matrix, i.e., the mapping matrix $C$. Suppose that at time $t$, human subject's (implicit) estimate of matrix $C$ is $\hat{C}(t)$. Correspondingly, the estimated change in cursor position $\widehat{\delta \bs x} \in \real^n$ for a change in finger joint angles $\delta \bs q \in \real^m$ is
\begin{equation}
    \widehat{\delta \bs x} = \hat{C}  \delta \bs q,     \label{ref_system}
\end{equation}
where $\hat{C} = \hat{W} \Phi$, and $\hat{W}(t) \in \real^{n\times h}$ is the parameter estimate matrix, which determines the weight the human subject assigns to each synergy. 
It follows from (\ref{system}) and (\ref{ref_system}) that the estimation error
\begin{equation}
    \bs \epsilon = \delta \bs x -\widehat{ \delta \bs x} = -\tilde{W} \Phi  \delta \bs q,    \label{est_error}
\end{equation}
where $\tilde{W}(t) = \hat{W}(t) - W \in \real^{n\times h}$ is the parameter estimation error.

It is well known in adaptive control literature that the gradient descent for $\hat{W}$  on $\frac{1}{2} \norm{\bs \epsilon}^2$ leads to estimation error converging to zero. Accordingly, we select
\begin{align}\label{fw_dyn}
    \dot{\hat{W}} &= -\gamma \nabla_{\hat{W}} \frac{1}{2}\norm{\bs \epsilon}^2 = \gamma \bs \epsilon  \delta \bs q\tran \Phi\tran,
\end{align}
where $\gamma>0$ is the forward learning rate, $\dot{\hat{W}} \in \real^{n\times h}$, and $(.)\tran$ represents the transpose. We posit \eqref{fw_dyn}, with $\gamma$ as a tunable parameter, as a model for human forward learning dynamics. This model is consistent with the error-based human motor learning paradigm and similar models have been used in literature; see e.g., \cite{pierella2019dynamics, herzfeld2014memory}.

\subsection{Adaptive Control-based Model of Inverse Learning Dynamics}
The inverse learning concerns identifying coordinated finger joint movements that drive the cursor to the desired position. If the mapping matrix $C$ is known, then the following proportional control will drive the cursor to the desired position $\supscr{x}{des}$
\begin{align}
     \supscr{\bs u}{des} = k_P C^{\dagger} (\supscr{\bs x} {des} - \bs x) = k_P C^{\dagger} \bs {e_x}, \label{fb_control}
\end{align}
where $C^{\dagger}$ is the pseudo-inverse of $C$, $\bs x$ is the current cursor position, and $k_P>0$ is the (scalar) proportional gain.
We refer to the difference between the desired and current cursor position as the \textit{Reaching Error (RE)} and denote it by $\bs {e_x}$.

Since the mapping matrix $C$ is not known, $C^{\dagger}$ in \eqref{fb_control} can be replaced by $\hat{C}^{\dagger}$, where $\hat{C}$ is computed using forward learning model. It may not be reasonable to assume that humans can compute the pseudo-inverse of their implicit estimate $\hat{C}$. Therefore, we postulate that human subjects determine control input \eqref{fb_control} through gradient-descent based learning. To this end, we let
\begin{align}
    \dot{{\bs u}} = - \eta \nabla_{{\bs u}} &\left(\frac{1}{2} \norm{\dot{\hat{\bs x}} - k_P \bs {e_x}}^2 + \frac{\mu}{2}\norm{ {\bs u}}^2\right),  \notag
\end{align}
where ${\bs u} \in \real^m$, and $\hat{\bs x}$ is the estimated cursor position based on subject's estimate of mapping matrix $\hat{C}(t)$. Here $\eta>0$ is the inverse learning rate, and $\mu>0$ is the regularizer weight. Note that these learning dynamics for $\bs u$ minimizes the difference between the estimated and desired cursor velocity subject to a regularization constraint for $\bs u$.
Upon simplification, these dynamics reduce to
\begin{align}
    \nabla_{ {\bs u}} &\left(\frac{1}{2} \norm{\dot{\hat{\bs x}} - k_P \bs {e_x}}^2 + \frac{\mu}{2}\norm{ {\bs u}}^2\right) \notag \\
    &= \nabla_{ {\bs u}} \left(\frac{1}{2} \norm{\hat{C}  {\bs u} - k_P \bs {e_x}}^2 + \frac{\mu}{2}\norm{ {\bs u}}^2\right) \notag \\
    &= \hat{C}\tran (\hat{C} {\bs u} - k_P \bs {e_x}) + \mu  {\bs u} \notag \\
    &= ( \hat{C}\tran \hat{C} + \mu I ){\bs u} - k_P\hat{C}\tran \bs {e_x},
\end{align}
which results in
\begin{align}
     \dot{{\bs u}} = &- \eta \left( (\hat{C}\tran \hat{C} + \mu I ){\bs u} - k_P\hat{C}\tran \bs {e_x} \right) \notag \\
                       = &- \eta \left( (\Phi\tran \hat{W}\tran \hat{W} \Phi + \mu I ){\bs u} - k_P \Phi\tran \hat{W}\tran \bs {e_x} \right). \label{inv_dyn}
\end{align}
It should be noted that we do not anticipate that human subjects can explicitly compute RHS of \eqref{inv_dyn}. However, we assume that they can implicitly determine RHS of \eqref{inv_dyn} by finding the direction of steepest descent for $\frac{1}{2} \norm{\dot{\hat{\bs x}} - k_P \bs {e_x}}^2 + \frac{\mu}{2}\norm{ {\bs u}}^2$. We posit \eqref{inv_dyn}, with $\eta$ as a tunable parameter, as a model for human inverse learning dynamics. Additionally, the regularization parameter $\mu$ determines the optimality of the learnt $\bs u$.

\begin{definition}[{Persistently Exciting Signal}]
    A signal vector $\omega(t)$ is persistently exciting if there exist $\alpha_1, \alpha_2, \delta > 0$ such that
    \begin{align}
        \alpha_2 I \geq \int_{t_0}^{t_0+\delta} \omega(\tau) \omega(\tau)\tran \mr{d}\tau \geq \alpha_1I, \textup{for all } t_0 \geq 0. \notag
    \end{align}
\end{definition}
\begin{assumption}[{Properties of the Input}] \label{assump1}
    $\delta \bs q, \dot{\delta \bs q} \in \mc L_{\infty}$, and $\delta \bs q$ is persistently exciting.
\end{assumption}

To ensure that $\delta \bs q$ is persistently exciting, a noise $\xi$ is added to the inverse dynamics (\ref{inv_dyn}) giving
\begin{align} \label{noisy_u}
 \!\!\!   \dot{{\bs u}} = &- \eta \left( (\Phi\tran \hat{W}\tran \hat{W} \Phi + \mu I ){\bs u} - k_P \Phi\tran \hat{W}\tran \bs {{e_x}}\right) + {\xi} ,
\end{align}
where $\alpha_2 T \geq \int_{t_0}^{t_0 + T} \xi(\tau) \xi(\tau)\tran \mr{d} \tau \geq \alpha_1 T$, for all $t_0 \geq 0$, and some $\alpha_1, \alpha_2, T > 0$.
\begin{remark}
Since $\delta \bs q$ is the filtered change in finger joint angles, it is reasonable to assume $\delta \bs q$ and its derivative to be bounded. Additionally, since humans use exploratory noise for motor learning~\cite{pierella2019dynamics}, assuming $\delta \bs q$ to be persistently exciting is also reasonable.
\end{remark}
We refer to \eqref{delta_u}, \eqref{fw_dyn}, and \eqref{noisy_u} together as the HML dynamics model.

\section{Analysis of the Adaptive Control-based model of HML Dynamics} \label{section4}
We now analyze the proposed model of HML in equations \eqref{delta_u}, \eqref{fw_dyn}, and \eqref{noisy_u}.
The motor learning literature suggests that the forward learning dynamics evolve on a faster timescale than the inverse learning dynamics; see \cite{yavari2013fast} for plausible neural mechanisms underlying this timescale separation. Consistent with the literature, our fits in Section \ref{section5.2} suggests $k_P \ll \eta \ll \gamma$, i.e., $\delta \bs q$ and $\hat{W}$ dynamics evolve on a faster timescale.
%
Decomposing equations  \eqref{base_dyn} and \eqref{ref_system} into $n$ equations (one each for cursor motion along axes), we obtain $\hat{\delta \bs x}_i=\hat{C}_i \delta \bs q$ for all $i\in\until{n}$. Similar expression can be written for Eq.~\eqref{base_dyn}. We analyze each of these equations separately, and for ease of notation we will drop the index $i$.
{
Defining the slow variable $\bar{e}_x = k_P e_x \in \real$, we get
}
\begin{equation}\label{eq:exdot}
 \dot{\bar{e}}_x = k_P(\dot{x} \supscr{} {des} - \dot{x}) = -k_P C {\bs u}= -k_P W\Phi {\bs u}.
\end{equation}

Let $\map{f_1}{\real^m \times \real \times \real^{1\times h}}{\real^m}, \ \map{f_2}{\real^m}{\real}, \ \map{f_3}{\real^m \times \real^m}{\real^m},\ \map{g}{\real^{1\times h} \times \real^m}{\real^{1\times h}}$ be defined by
\begin{align}
    \begin{split}
        f_1({\bs u}, \bar{e}_x, \hat{W}) = &-( (\hat{C}\tran \hat{C} + \mu I){\bs u} - \hat{C}\tran \bar{e}_x )\\
        f_2(\bs u) = &-W \Phi {\bs u}  \notag\\
        f_3(\delta \bs q, \bs u) = &-\delta \bs q + \bs u/a  \notag\\
        g(\tilde{W}, \delta \bs q) = &-\tilde{W} \Phi \delta \bs q \delta \bs q\tran \Phi\tran.
    \end{split}
\end{align}
Then the model of HML dynamics \eqref{delta_u}, \eqref{fw_dyn}, \eqref{noisy_u}, and \eqref{eq:exdot} can be equivalently written as
\begin{subequations} \label{sspm_u}
\begin{align}
       \dot{{\bs u}} &= \eta~f_1({\bs u}, \bar{e}_x, \hat{W}) + \xi\\
       \dot{\bar{e}}_x &= k_P~f_2(\bs u)\\
       \dot{\delta \bs q} &= a~f_3(\delta \bs q, \bs u) \\
        \dot{\hat{W}} &= \gamma~g(\tilde{W}, \delta \bs q). 
\end{align}
\end{subequations}

Rewriting system (\ref{sspm_u}) in the new timescale $t \mapsto k_P t$, and employing the change of variables $\tilde{W} = \hat{W} - {W},$ where $W$ is the vector of the true weights on synergies, we get the singularly perturbed system defined by
\begin{subequations} \label{sspm_y}
\begin{align} 
        \varepsilon_u \dot{{\bs u}} &= f_1({\bs u}, \bar{e}_x, \tilde{W}) + \xi/\eta\\
        \dot{\bar{e}}_x &= f_2(\bs u)\\
        \varepsilon_{\delta}\dot{\delta \bs q} &= f_3(\delta \bs q, \bs u) \\
        \varepsilon_W \dot{\tilde{W}} &= g(\tilde{W}, \delta \bs q),
\end{align}
\end{subequations}
where $\varepsilon_u = k_P/\eta$, $\varepsilon_{\delta} = k_P/a$, $\varepsilon_W = k_P/\gamma$, and we have used $\dot{\tilde{W}} = \dot{\hat{W}}$.
It can be verified that $({\bs u}, \bar{e}_x, \delta \bs q, \tilde{W}) =(0,0,\bs u/a,0)$ is an isolated equilibrium of system (\ref{sspm_y}) in the absence of noise $\xi$.
Moreover, the functions $f_1, f_2, f_3, g$ are locally Lipschitz and their partial derivatives up to the second-order are bounded in their respective domains containing the origin.
\subsection{Slower Timescale Inverse Learning Dynamics}
Ignoring the noise term in $\bs u$ dynamics, the reduced system associated with (\ref{sspm_y}) is
\begin{subequations} \label{RS}
\begin{align}
        \varepsilon_u\dot{{\bs u}} = &{f}_{1}({\bs u}, \bar{e}_x, 0) \notag \\
        = &-\left( (\Phi\tran W\tran W \Phi + \mu I){\bs u} - \Phi\tran W\tran \bar{e}_x \right) \\
       \dot{\bar{e}}_x = &{f}_{2}(\bs u) = -W \Phi {\bs u}. 
\end{align}
\end{subequations}

\begin{lemma}[\bit{Stability of the Reduced System}]\label{lemma1}
   For the reduced system \eqref{RS}, the equilibrium point at the origin is globally asymptotically stable.
\end{lemma}
\begin{proof}
Consider the radially unbounded Lyapunov function 
    \begin{equation}
        V({\bs u}, \bar{e}_x) = \frac{\varepsilon_u}{2} {\bs u}\tran {\bs u} + \frac{1}{2 k_P}{\bar{e}_x}^{2}. \notag
    \end{equation}
    The time-derivative of $V$ along the reduced system (\ref{RS}) is
    \begin{align}  \label{V_dot}
        \dot{V} = & \varepsilon_u{\bs u}\tran \dot{{\bs u}} + \frac{1}{k_P} \bar{e}_x\dot{\bar{e}}_x \notag\\
                = & -{\bs u}\tran (\Phi\tran W\tran W \Phi + \mu I) {\bs u} + {\bs u}\tran \Phi\tran W\tran \bar{e}_x \notag\\
                  &- W \Phi {\bs u} \bar{e}_x \notag \\
                = & -{\bs u}\tran (C\tran C + \mu I){\bs u} + ({\bs u}\tran C\tran - C {\bs u}) \bar{e}_x \notag \\
                \leq & -\alpha_1 \norm{{\bs u}}^2,
    \end{align}
    where ${\alpha}_{1} = \subscr{\lambda}{min}(C\tran C + \mu I) > 0$ is the smallest eigenvalue of the matrix in the argument, and 
    ${\bs u}\tran C\tran = C {\bs u}$ since it is a scalar.  
    The set $S=\setdef{({\bs u}, \bar{e}_x)}{\dot{V} = 0}$ does not contain any trajectory of the reduced system except for $({\bs u}, \bar{e}_x) = (0, 0)$. Hence, using LaSalle's invariance principle, the origin of reduced system \eqref{RS} is globally asymptotically stable.
\end{proof}
\subsection{Faster Timescale Forward Learning Dynamics}
The boundary-layer system associated with \eqref{sspm_y} is defined in the fast timescale $\tau_W = t/\varepsilon_W$ as
\begin{subequations}    \label{BLS}
\begin{align}
    \frac{d\tilde{W}}{d\tau_W} &= g(\tilde{W}, \delta \bs q)\\
    \frac{d\delta \bs q}{d\tau_W} &= \frac{\varepsilon_W}{\varepsilon_{\delta}}~f_3(\delta \bs q, \bs u),
\end{align}
\end{subequations}
where the slow-varying states ${\bs u}$, $\bar{e}_x$ are considered frozen.

\begin{lemma}[\bit{Stability of the Boundary Layer System}]\label{lemma2}
Under Assumption \ref{assump1}, for the boundary layer system \eqref{BLS}, the equilibrium point at $(0, \bs u/a)$ is globally exponentially stable.
\end{lemma}
\begin{proof}
The lemma can be proved using standard adaptive control techniques~\cite{ioannou1996robust}, which leverage the radially unbounded Lyapunov function
   $V_b(\tilde{W}, \delta \bs q) = \tilde{W} \tilde{W}\tran + \frac{\varepsilon_{\delta}}{\varepsilon_W}\left(\delta \bs q -  \frac{\bs u}{a}\right)\tran \left(\delta \bs q - \frac{\bs u}{a}\right)$
and Assumption \ref{assump1}.
\end{proof}
\begin{figure*}[ht!]
    \centering
    \begin{subfigure}{0.25\linewidth}
	    \centering
        \includegraphics[width=1\linewidth, height=1.1\linewidth, keepaspectratio]{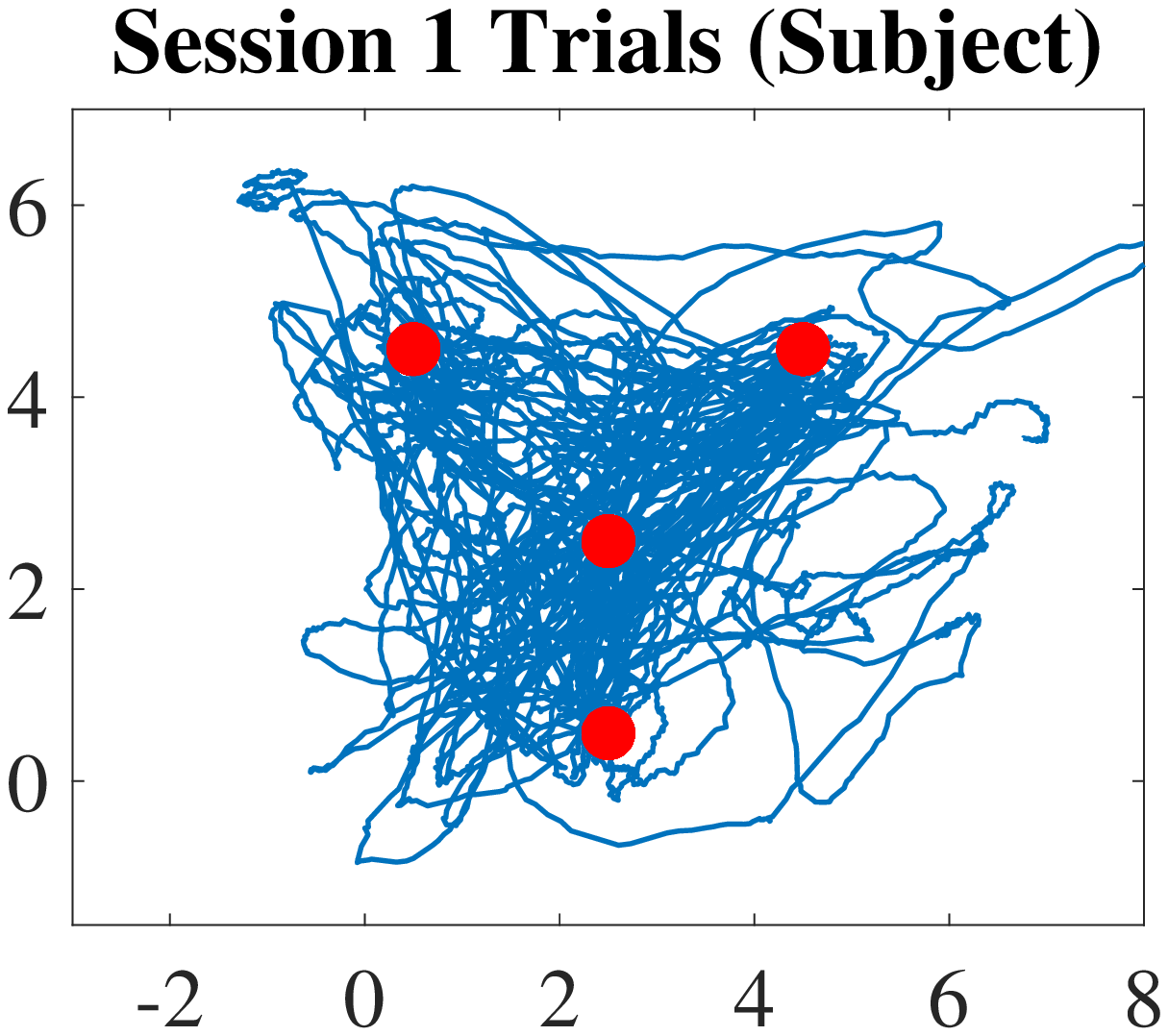}
        \caption{}
        \label{a}
    \end{subfigure}\hspace{-1.7em}%
    ~
    \begin{subfigure}{0.25\linewidth}
	    \centering
        \includegraphics[width=1\linewidth, height=1.1\linewidth, keepaspectratio]{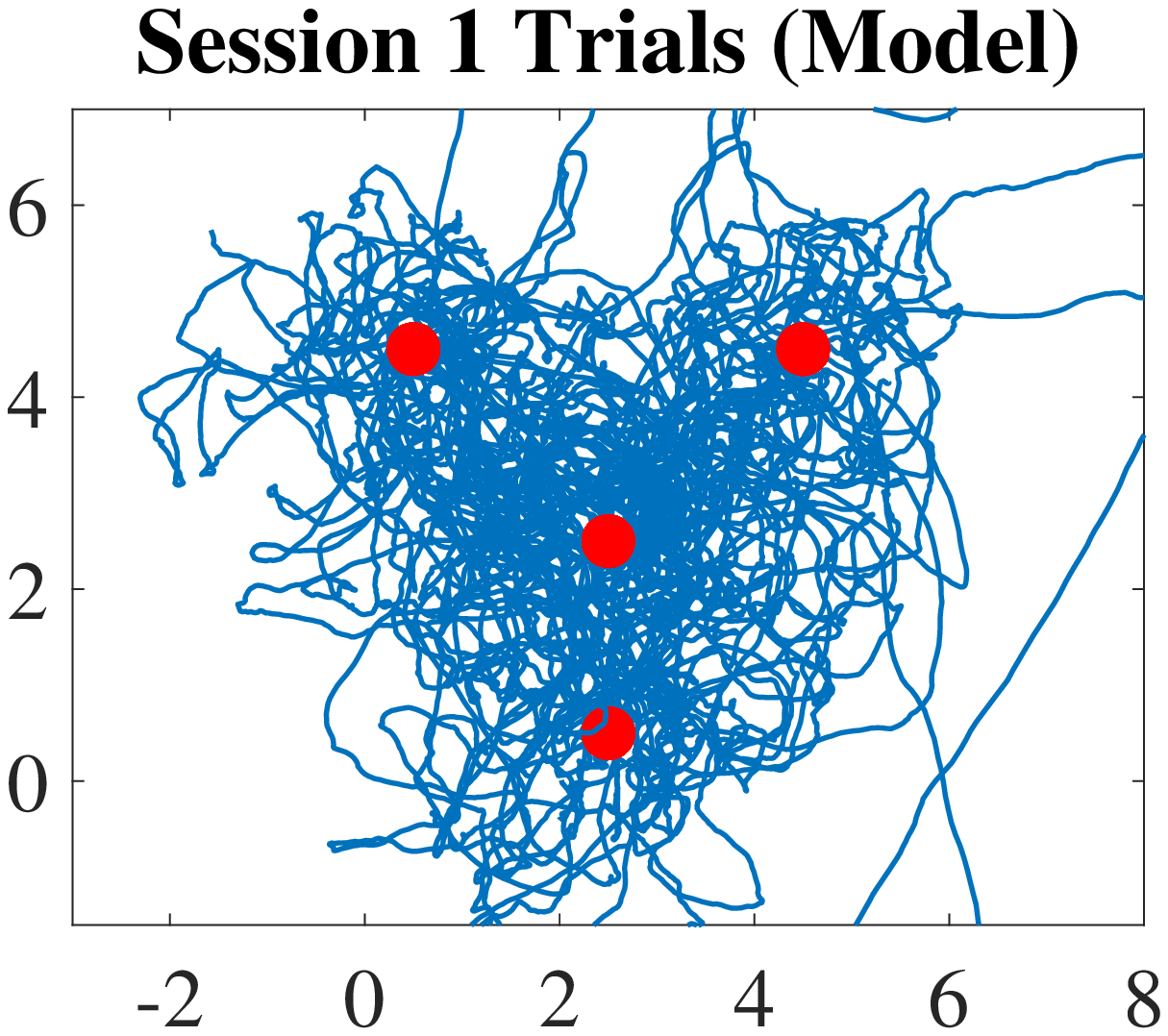}
        \caption{}
        \label{b}
    \end{subfigure}\hspace{-1.7em}%
    ~
    \begin{subfigure}{0.25\linewidth}
	    \centering
        \includegraphics[width=1\linewidth, height=1.1\linewidth, keepaspectratio]{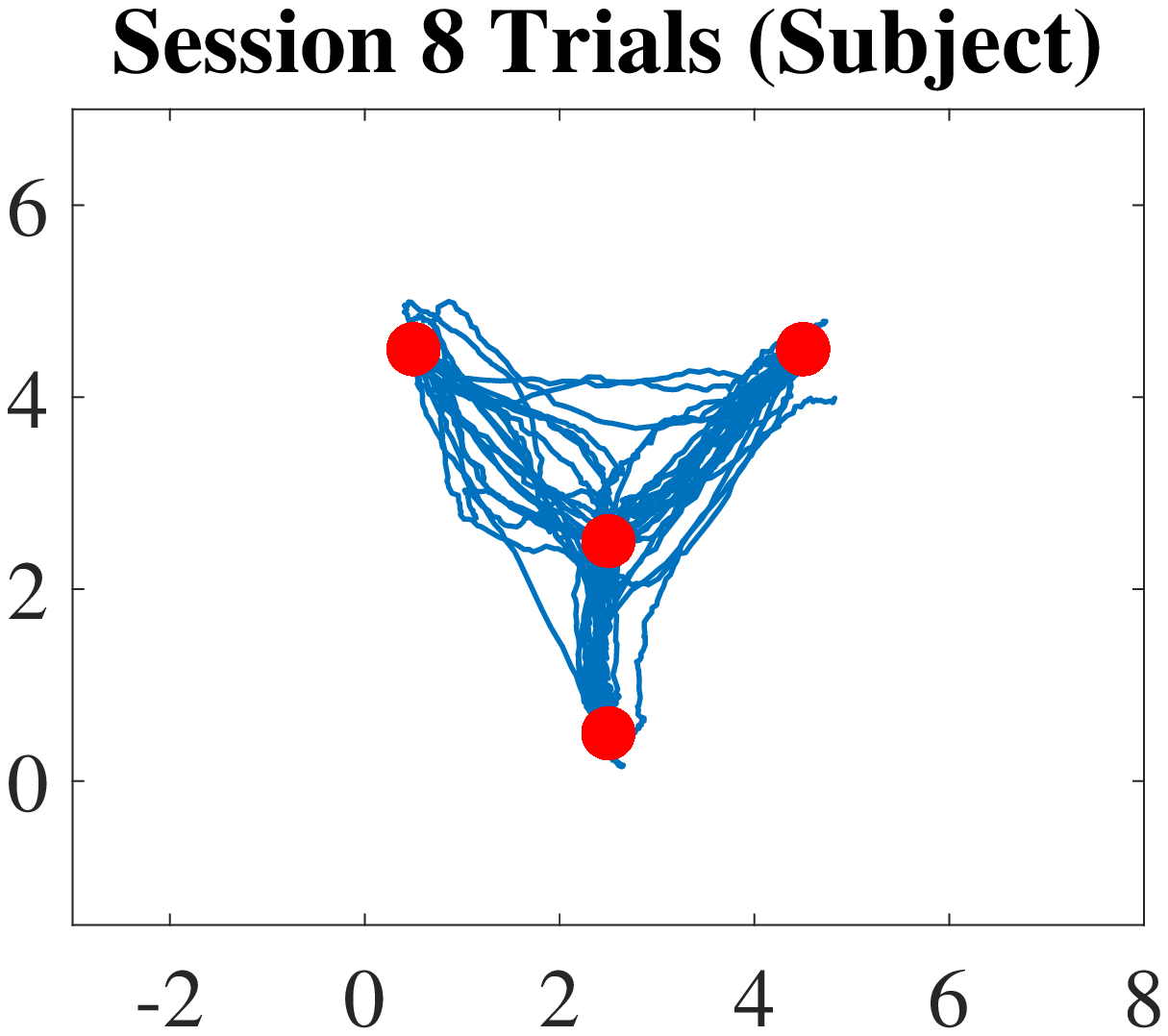}
        \caption{}
        \label{c}
    \end{subfigure}\hspace{-1.7em}%
    ~
    \begin{subfigure}{0.25\linewidth}
	    \centering
        \includegraphics[width=1\linewidth, height=1.1\linewidth, keepaspectratio]{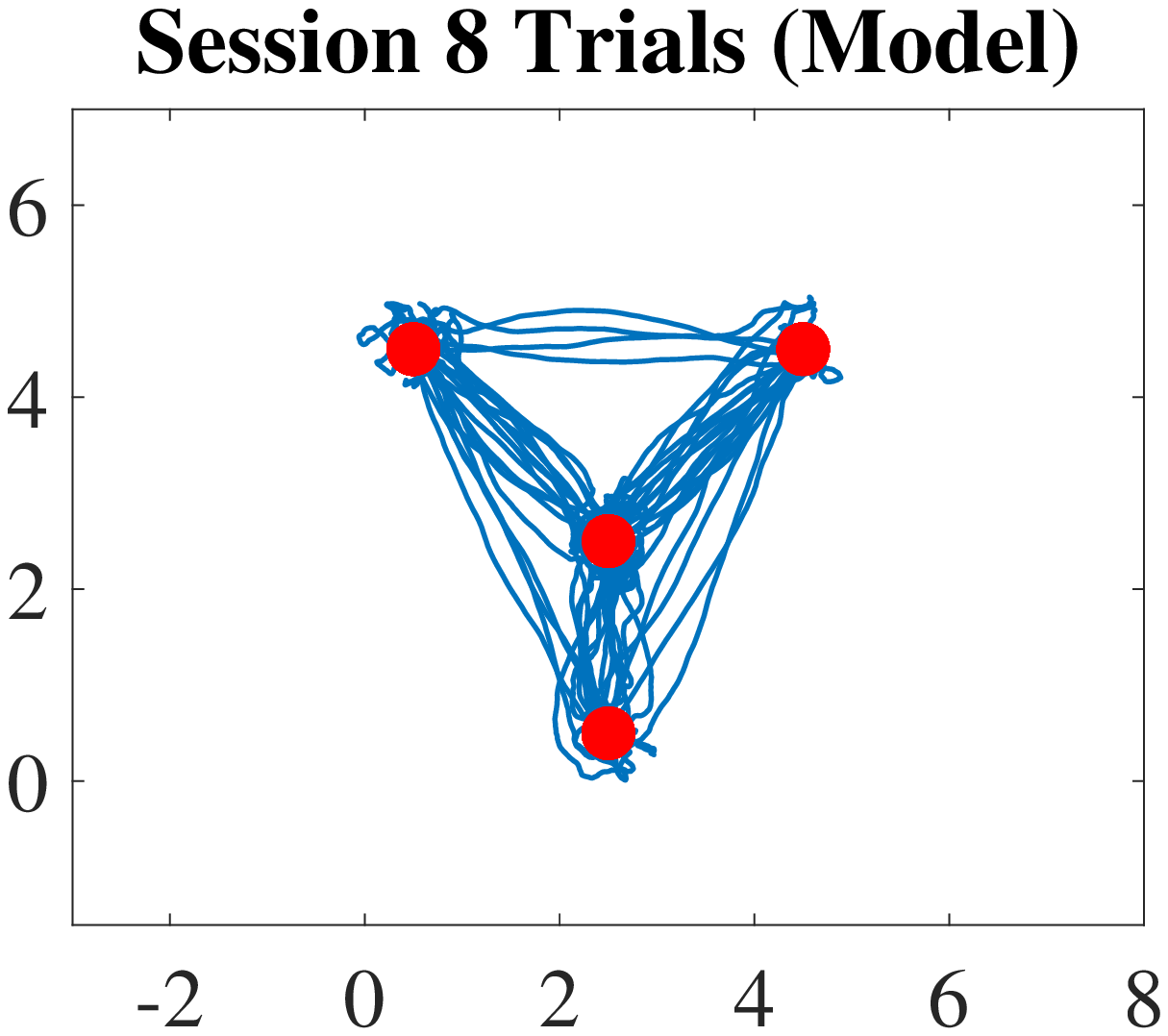}
        \caption{}
        \label{d}
    \end{subfigure}
    \caption{\textbf{Cursor Trajectories:} Cursor trajectory data from human experiments and the fitted model.}
    \label{fig:hand_traj}
\end{figure*}
\subsection{Coupled Forward-Inverse Motor Learning Dynamics}

We now show the stability and convergence of our proposed HML model using singular perturbation arguments.
\begin{theorem}[\bit{Stability of the HML Model~\eqref{sspm_u}}]
For $\varepsilon_W \in (0, \varepsilon^*)$, for $\varepsilon^*$ sufficiently small, and under Assumption \ref{assump1}, the trajectories of HML model \eqref{sspm_u} asymptotically converge to an $O(\|\xi\|)$ neighborhood of the origin.
\end{theorem}

\begin{proof}With the inclusion of persistently exciting noise $\xi$, the reduced system associated with \eqref{sspm_y} is an $O(\|\xi\|)$ perturbation of the system \eqref{RS}.
Using similar techniques as in \cite[Theorem 11.3]{HKK:02}, the Lyapunov functions in Lemmas \ref{lemma1} and \ref{lemma2} can be used to construct a composite Lyapunov function. Then, the theorem can be established using this composite Lyapunov function and arguments similar to~\cite[Lemma 9.3]{HKK:02}.
\end{proof}
\section{Model Fitting and Comparison with Experimental Data} \label{section5}
In this section, we discuss the methods used in designing the BoMI mapping matrix and finding the model parameters by fitting data from human subject experiments. We then show the performance of our proposed HML model by comparing it with the experimental data.
\subsection{Designing the Forward BoMI Mapping Matrix} \label{section5.1.1}
To design the forward BoMI mapping matrix, Ranganathan \etal~\cite{ranganathan2013learning} collect hand posture data during a \emph{free finger exploration} phase, in which subjects freely choose arbitrary hand postures per their discretion. They perform Principal Component Analysis (PCA) on the collected hand posture data. The first two principal components account for more than $80\%$ of the variance in the joint movement data and are used in the mapping matrix $C$ to map the movements of hand finger joints to cursor movement in $x$ and $y$ direction. 
To reduce the dimension of the learning space, we construct a synergy matrix $\Phi$ using the first four principal components from the above PCA.
While our model can work with any number of synergies, 
we choose four synergies since they are sufficiently representative of the free finger movements (see Section \ref{section3}), and the higher number of synergies requires a higher exploration/excitation noise.
%
%
\subsection{Fitting HML Model to Experimental Data}\label{section5.2} 
We refer to the Euclidean norm of the cursor position at the end of movement and the target point as Reaching Error (RE). End of movement is defined as the time when the cursor reaches inside a 0.15 units radius circle of the target or after 2 seconds of the start of the movement, whichever is earlier. It can be numerically verified that for the HML model, RE follows an exponential trend with the rate of convergence $\eta$. Therefore, for the experimental data, RE is computed by averaging over similar trials (trials with same start and end points) across sessions, followed by smoothing using a 10-trials moving window averaging. The learning rate $\eta$ in the inverse learning dynamics (\ref{inv_dyn}) is determined by fitting RE as an exponential function of trial index $k$, i.e., $\text{RE} = \alpha \exp(-\eta k) + c$. For our experimental data, the estimate of $\eta$ is $0.04522$ with the goodness-of-fit parameter R$^2 = 0.9019$. With the estimated inverse learning rate, we select the forward learning rate $\gamma$ as a minimizer of the norm of the difference between finger joint trajectories from human data and the model. To this end, we perform a grid search over $\gamma \in [0, 10]$, and find the optimum value of $\gamma$ to be equal to $0.262$.

Other variables are heuristically selected as: $k_P\sim O(10^{-3})$, decreasing over sessions, and $\mu =0.3$. $\xi$ is selected as a zero-mean white noise with variance $\subscr{S}{session}(0.01+e^{-0.1 t})$, where $t$ is the time index within the session and $\subscr{S}{session}$ is decreasing over sessions. For these parameters, the HML model is compared with the experimental data in the next section.

\begin{figure}
    \centering
    \begin{subfigure}{0.5\linewidth}
	    \centering
        \includegraphics[width=1\linewidth, height=1.1\linewidth, keepaspectratio]{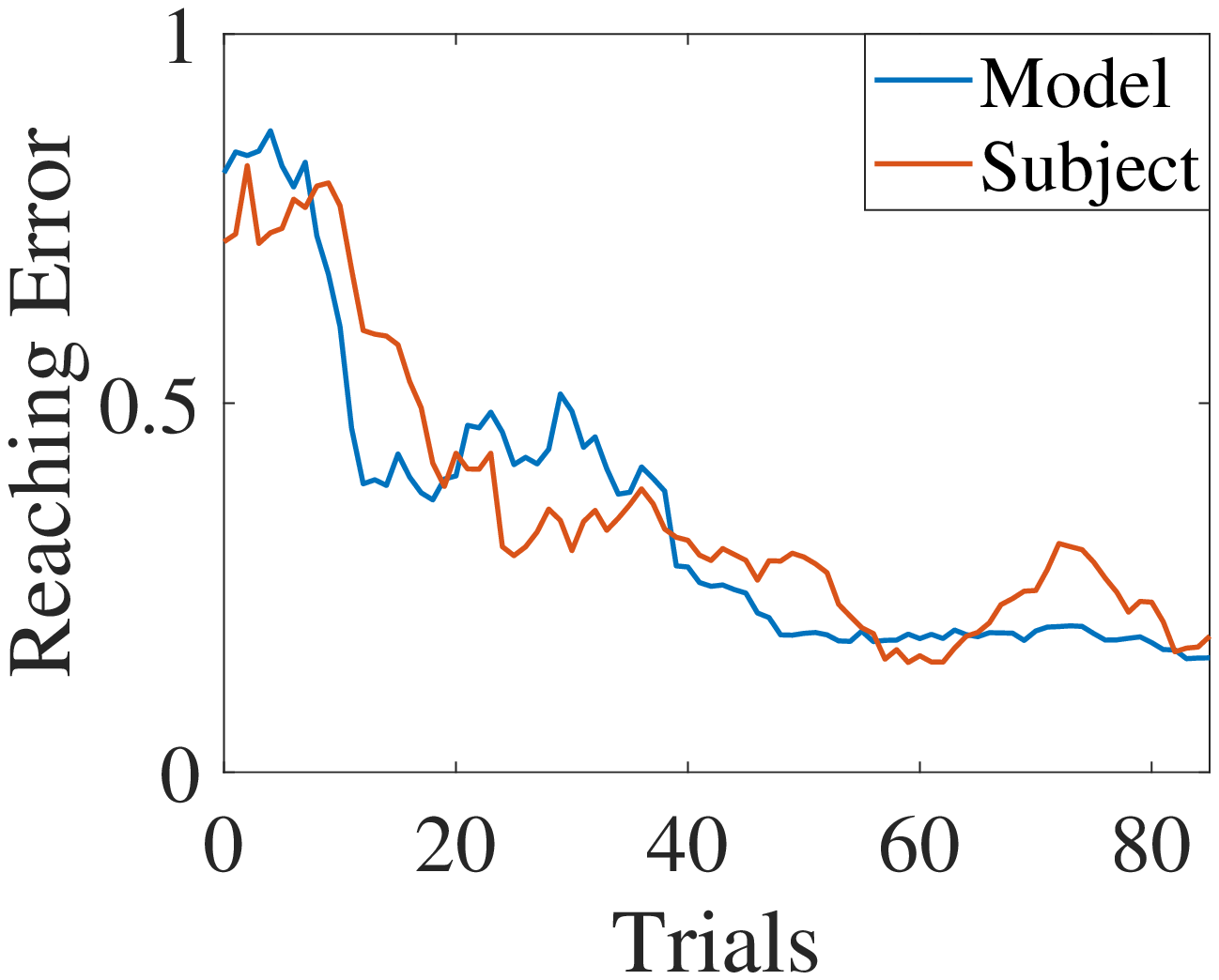}
        \caption{}
        \label{RE}
    \end{subfigure}\hspace{-1.7em}%
    ~
    \begin{subfigure}{0.5\linewidth}
	    \centering
        \includegraphics[width=1\linewidth, height=1.1\linewidth, keepaspectratio]{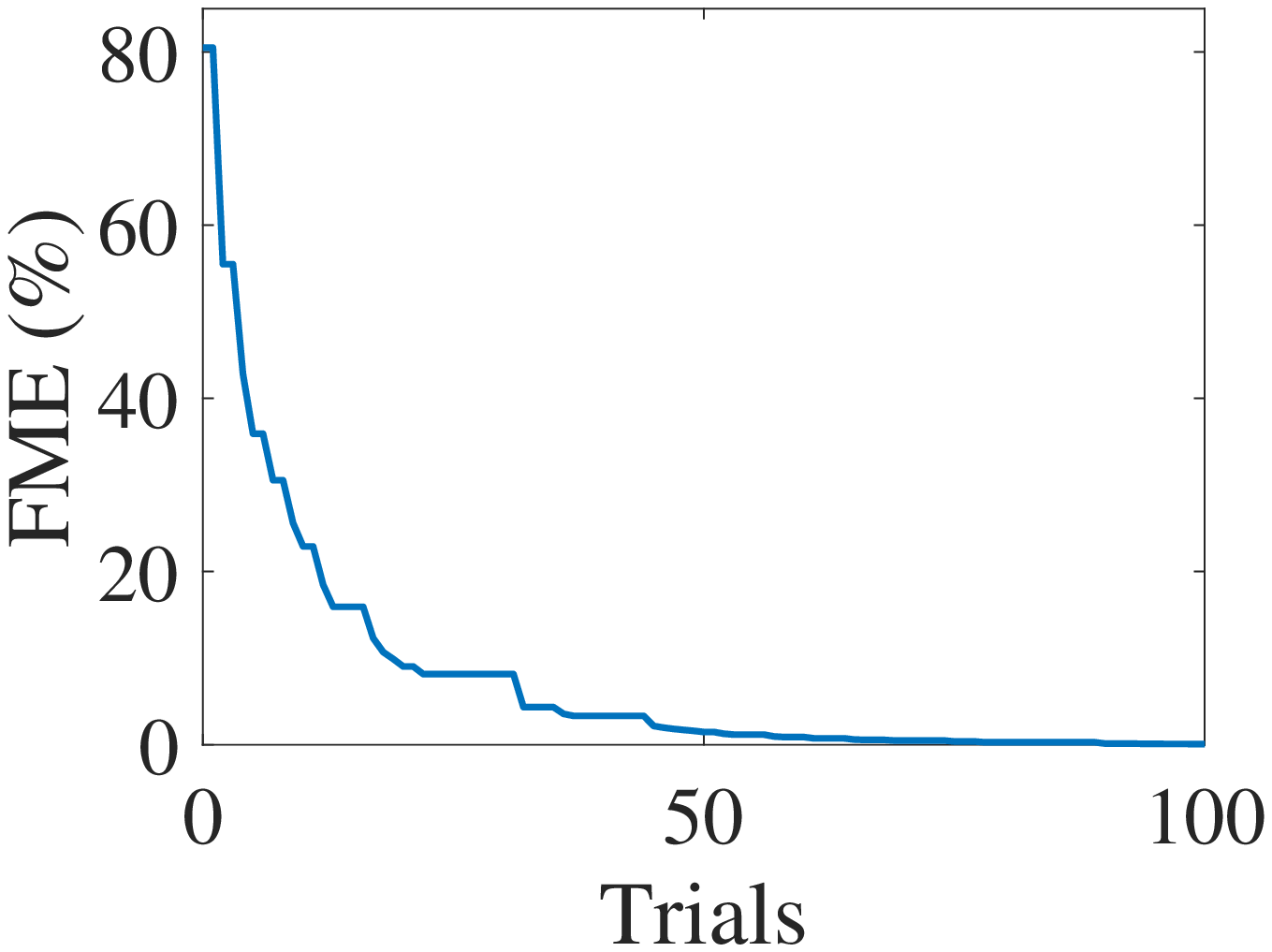}
        \caption{}
        \label{FME}
    \end{subfigure}
    \caption{\textbf{Performance Measures:} (a) RE for the subject (red) and for the fitted model (blue) as a function of trials. (b)  Evolution of FME for the subject as a function of trials, showing only the first 100 trials.}
    \label{fig:FME_RE}
\end{figure}

\subsection{Performance Measures: Comparing with Human Hand Data}\label{section5.perfmeas}
We simulate the HML model with the parameters estimated above. The model is simulated at $100$ Hz, while the data glove samples finger joint positions at $50$ Hz. Therefore, the trajectory data from the simulation model is sub-sampled at the same time indices as the human data for a direct comparison. We use Forward Model Error (FME) \cite{pierella2019dynamics} to quantify the convergence of the subject's estimate of forward mapping, $\hat{C} = \hat{W} \Phi$, to the actual forward mapping matrix, $C= W\Phi$, which is defined by
\begin{align}
    \text{FME}_{k} = \frac{\norm{C - \hat{C}}}{\norm{C}} = \frac{\norm{W - \hat{W}}}{\norm{W}}. \notag
\end{align}
\begin{figure}
    \centering
    \includegraphics[width=\linewidth]{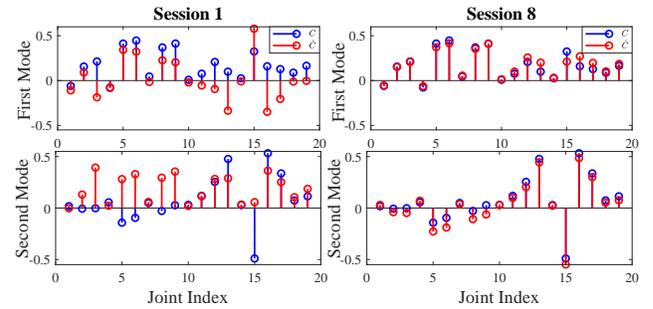}
    \caption{\textbf{Modes of Mapping Matrix}: Evolution of SVD modes of the forward mapping matrix from Session 1 to Session 8.}
    \label{fig:modes}
    \vspace{-0.1in}
\end{figure}
Fig.~\ref{fig:hand_traj} compares the actual trajectories obtained from the human data with that coming out of our proposed model. Results show that our model can mimic the motor learning of human hand motions very closely. Noisy trajectories in the beginning are due to high exploration noise in the initial trials; they look straighter as the trials progress and the subject/model learns the mapping.
The evolution of both FME and RE with the trials is shown in Fig.~\ref{fig:FME_RE}. Both FME and RE decrease with trials and the behavior of the fitted model is consistent with the experimental data. Fig.~\ref{fig:modes} compares the SVD modes of learned forward mapping $\hat{C}$ with the SVD modes of actual mapping matrix $C$. Results show that the modes of learned mapping are more consistent with that of the actual mapping matrix for the later sessions, and thus $\hat{C}$ converges to $C$ with high accuracy.

\section{Conclusion and Future Directions}\label{section6}
In this paper, we propose an HML model that can mimic human motor learning in high-dimensional spaces. We leverage the low-dimensional key motor synergies that drive the movement of human finger joints in high-dimensions and build up on the internal model theory of motor control to obtain a mathematical model of motor learning. Using singular perturbation arguments, we establish that the proposed model of learning dynamics is asymptotically stable. We further establish the conditions under which human estimate of the forward mapping converges to the actual forward mapping. We discuss techniques to fit the proposed model to the experimental data and show that the model captures experimental behavior well.

This model can be leveraged for model-based investigation of motor rehabilitation and for designing strategies for effective motor learning assistance. For example, the HML model can be leveraged for optimal scheduling of motor rehabilitation training. The HML model could potentially help in designing a bi-directional human-robot learning framework that would allow the adaptation of rehabilitation strategies to subject-specific needs and requirements. Furthermore, robotic rehabilitation for high-dimensional motor systems has gained prominence \cite{rashid2019wearable,zhang2014design, castiblanco2021assist, agarwal2017subject, agarwal2019framework}, and we believe that the HML model developed in this paper can be used with such rehabilitation systems to enable bi-directional human-robot interaction.

\bibliographystyle{IEEEtran}
\bibliography{motor_learning,mybib} 

\begin{thebibliography}{10}
\providecommand{\url}[1]{#1}
\csname url@samestyle\endcsname
\providecommand{\newblock}{\relax}
\providecommand{\bibinfo}[2]{#2}
\providecommand{\BIBentrySTDinterwordspacing}{\spaceskip=0pt\relax}
\providecommand{\BIBentryALTinterwordstretchfactor}{4}
\providecommand{\BIBentryALTinterwordspacing}{\spaceskip=\fontdimen2\font plus
\BIBentryALTinterwordstretchfactor\fontdimen3\font minus
  \fontdimen4\font\relax}
\providecommand{\BIBforeignlanguage}[2]{{%
\expandafter\ifx\csname l@#1\endcsname\relax
\typeout{** WARNING: IEEEtran.bst: No hyphenation pattern has been}%
\typeout{** loaded for the language `#1'. Using the pattern for}%
\typeout{** the default language instead.}%
\else
\language=\csname l@#1\endcsname
\fi
#2}}
\providecommand{\BIBdecl}{\relax}
\BIBdecl

\bibitem{mozaffarian2015executive}
D.~Mozaffarian, E.~J. Benjamin, A.~S. Go, D.~K. Arnett, M.~J. Blaha,
  M.~Cushman, S.~De~Ferranti, J.-P. Despr{\'e}s, H.~J. Fullerton, V.~J. Howard
  \emph{et~al.}, ``Heart disease and stroke statistics - 2015 update: A report
  from the {A}merican {H}eart {A}ssociation,'' \emph{Circulation}, vol. 131,
  no.~4, pp. 434--441, 2015.

\bibitem{g1999long}
J.~G.~Broeks, G.~Lankhorst, K.~Rumping, and A.~Prevo, ``The long-term outcome
  of arm function after stroke: {R}esults of a follow-up study,''
  \emph{Disability and Rehabilitation}, vol.~21, no.~8, pp. 357--364, 1999.

\bibitem{lawrence2001estimates}
E.~S. Lawrence, C.~Coshall, R.~Dundas, J.~Stewart, A.~G. Rudd, R.~Howard, and
  C.~D. Wolfe, ``Estimates of the prevalence of acute stroke impairments and
  disability in a multiethnic population,'' \emph{Stroke}, vol.~32, no.~6, pp.
  1279--1284, 2001.

\bibitem{cruz2005kinetic}
E.~Cruz, H.~Waldinger, and D.~Kamper, ``Kinetic and kinematic workspaces of the
  index finger following stroke,'' \emph{Brain}, vol. 128, no.~5, pp.
  1112--1121, 2005.

\bibitem{lang2003differential}
C.~E. Lang and M.~H. Schieber, ``Differential impairment of individuated finger
  movements in humans after damage to the motor cortex or the corticospinal
  tract,'' \emph{Journal of Neurophysiology}, vol.~90, no.~2, pp. 1160--1170,
  2003.

\bibitem{li2003effects}
S.~Li, M.~L. Latash, G.~H. Yue, V.~Siemionow, and V.~Sahgal, ``The effects of
  stroke and age on finger interaction in multi-finger force production
  tasks,'' \emph{Clinical Neurophysiology}, vol. 114, no.~9, pp. 1646--1655,
  2003.

\bibitem{zhou2016learning}
S.-H. Zhou, J.~Fong, V.~Crocher, Y.~Tan, D.~Oetomo, and I.~Mareels, ``Learning
  control in robot-assisted rehabilitation of motor skills--a review,''
  \emph{Journal of Control and Decision}, vol.~3, no.~1, pp. 19--43, 2016.

\bibitem{tresch2006matrix}
M.~C. Tresch, V.~C. Cheung, and A.~d'Avella, ``Matrix factorization algorithms
  for the identification of muscle synergies: Evaluation on simulated and
  experimental data sets,'' \emph{Journal of Neurophysiology}, vol.~95, no.~4,
  pp. 2199--2212, 2006.

\bibitem{santello1998postural}
M.~Santello, M.~Flanders, and J.~F. Soechting, ``Postural hand synergies for
  tool use,'' \emph{Journal of Neuroscience}, vol.~18, no.~23, pp.
  10\,105--10\,115, 1998.

\bibitem{haith2013model}
A.~M. Haith and J.~W. Krakauer, ``Model-based and model-free mechanisms of
  human motor learning,'' in \emph{Progress in Motor Control}.\hskip 1em plus
  0.5em minus 0.4em\relax Springer, 2013, pp. 1--21.

\bibitem{shadmehr1994adaptive}
R.~Shadmehr and F.~A. Mussa-Ivaldi, ``Adaptive representation of dynamics
  during learning of a motor task,'' \emph{Journal of Neuroscience}, vol.~14,
  no.~5, pp. 3208--3224, 1994.

\bibitem{zhou2011human}
S.-H. Zhou, D.~Oetomo, Y.~Tan, E.~Burdet, and I.~Mareels, ``Human motor
  learning through iterative model reference adaptive control,'' in \emph{Proc.
  of IFAC World Congress}, 2011, pp. 2883--2888.

\bibitem{pierella2019dynamics}
C.~Pierella, M.~Casadio, F.~A. Mussa-Ivaldi, and S.~A. Solla, ``The dynamics of
  motor learning through the formation of internal models,'' \emph{PLoS
  Computational Biology}, vol.~15, no.~12, p. e1007118, 2019.

\bibitem{shadmehr2010error}
R.~Shadmehr, M.~A. Smith, and J.~W. Krakauer, ``Error correction, sensory
  prediction, and adaptation in motor control,'' \emph{Annual Review of
  Neuroscience}, vol.~33, pp. 89--108, 2010.

\bibitem{yavari2013fast}
F.~Yavari, F.~Towhidkhah, and M.~A. Ahmadi-Pajouh, ``Are fast/slow process in
  motor adaptation and forward/inverse internal model two sides of the same
  coin?'' \emph{Medical Hypotheses}, vol.~81, no.~4, pp. 592--600, 2013.

\bibitem{muratori2013applying}
L.~M. Muratori, E.~M. Lamberg, L.~Quinn, and S.~V. Duff, ``Applying principles
  of motor learning and control to upper extremity rehabilitation,''
  \emph{Journal of Hand Therapy}, vol.~26, no.~2, pp. 94--103, 2013.

\bibitem{krakauer2006motor}
J.~W. Krakauer, ``Motor learning: its relevance to stroke recovery and
  neurorehabilitation,'' \emph{Current Opinion in Neurology}, vol.~19, no.~1,
  pp. 84--90, 2006.

\bibitem{mosier2005remapping}
K.~M. Mosier, R.~A. Scheidt, S.~Acosta, and F.~A. Mussa-Ivaldi, ``Remapping
  hand movements in a novel geometrical environment,'' \emph{Journal of
  Neurophysiology}, vol.~94, no.~6, pp. 4362--4372, 2005.

\bibitem{sastry1990adaptive}
S.~Sastry, M.~Bodson, and J.~F. Bartram, \emph{Adaptive Control: Stability,
  Convergence, and Robustness}.\hskip 1em plus 0.5em minus 0.4em\relax
  Acoustical Society of America, 1990.

\bibitem{vinjamuri2014candidates}
R.~Vinjamuri, V.~Patel, M.~Powell, Z.-H. Mao, and N.~Crone, ``Candidates for
  synergies: linear discriminants versus principal components,''
  \emph{Computational Intelligence and Neuroscience}, vol. 2014, 2014.

\bibitem{liu2008contributions}
X.~Liu and R.~A. Scheidt, ``Contributions of online visual feedback to the
  learning and generalization of novel finger coordination patterns,''
  \emph{Journal of Neurophysiology}, vol.~99, no.~5, pp. 2546--2557, 2008.

\bibitem{herzfeld2014memory}
D.~J. Herzfeld, P.~A. Vaswani, M.~K. Marko, and R.~Shadmehr, ``A memory of
  errors in sensorimotor learning,'' \emph{Science}, vol. 345, no. 6202, pp.
  1349--1353, 2014.

\bibitem{ioannou1996robust}
P.~A. Ioannou and J.~Sun, \emph{Robust Adaptive Control}.\hskip 1em plus 0.5em
  minus 0.4em\relax PTR Prentice-Hall Upper Saddle River, NJ, 1996, vol.~1.

\bibitem{HKK:02}
H.~K. Khalil, \emph{Nonlinear Systems}, 3rd~ed.\hskip 1em plus 0.5em minus
  0.4em\relax Prentice Hall, 2002.

\bibitem{ranganathan2013learning}
R.~Ranganathan, A.~Adewuyi, and F.~A. Mussa-Ivaldi, ``Learning to be lazy:
  Exploiting redundancy in a novel task to minimize movement-related effort,''
  \emph{Journal of Neuroscience}, vol.~33, no.~7, pp. 2754--2760, 2013.

\bibitem{rashid2019wearable}
A.~Rashid and O.~Hasan, ``Wearable technologies for hand joints monitoring for
  rehabilitation: A survey,'' \emph{Microelectronics Journal}, vol.~88, pp.
  173--183, 2019.

\bibitem{zhang2014design}
F.~Zhang, L.~Hua, Y.~Fu, H.~Chen, and S.~Wang, ``Design and development of a
  hand exoskeleton for rehabilitation of hand injuries,'' \emph{Mechanism and
  Machine Theory}, vol.~73, pp. 103--116, 2014.

\bibitem{castiblanco2021assist}
J.~C. Castiblanco, I.~F. Mondragon, C.~Alvarado-Rojas, and J.~D. Colorado,
  ``Assist-as-needed exoskeleton for hand joint rehabilitation based on muscle
  effort detection,'' \emph{Sensors}, vol.~21, no.~13, p. 4372, 2021.

\bibitem{agarwal2017subject}
P.~Agarwal and A.~D. Deshpande, ``Subject-specific assist-as-needed controllers
  for a hand exoskeleton for rehabilitation,'' \emph{IEEE Robotics and
  Automation Letters}, vol.~3, no.~1, pp. 508--515, 2017.

\bibitem{agarwal2019framework}
------, ``A framework for adaptation of training task, assistance and feedback
  for optimizing motor (re)-learning with a robotic exoskeleton,'' \emph{IEEE
  Robotics and Automation Letters}, vol.~4, no.~2, pp. 808--815, 2019.

\end{thebibliography}

\end{document}